\documentclass[11pt,a4paper]{article}
\usepackage[margin=25mm]{geometry}
\usepackage{amsmath, amssymb, amsthm}
\usepackage{microtype}
\usepackage{hyperref}
\usepackage{cite}
\usepackage{enumitem}

\usepackage{empheq}
\usepackage{tikz}
\usepackage{subcaption}
\usepackage{comment}
\usepackage{tikz-cd}
\usetikzlibrary{arrows.meta,calc,positioning,decorations.markings,shapes.geometric}

\newtheorem{proposition}{Proposition}
\numberwithin{equation}{section}
\newcommand{\G}[2]{\Gamma(#1\leftarrow #2)}
\title{On the emergence of quantum mechanics from stochastic processes}
\author{Jason Doukas}
\date{}

\begin{document}

\maketitle

\begin{abstract}
The stochastic--quantum correspondence reinterprets quantum dynamics in terms of an underlying stochastic process on a configuration space. We generalize its lifting prescription, overcoming restrictions on the reachable density matrices without enlarging the configuration space or its associated Hilbert space. A compatibility relation and Kraus dictionary connect stochastic kernels to completely positive, trace-preserving (CPTP) maps: every finite-dimensional kernel admits a compatible lift, and these lifts collectively encompass all quantum channels as the kernels vary. We describe sufficiently regular composable CPTP dynamics through time-local Lindblad generators. We establish that every rooted quantum dynamical family admits a stochastic realizer on the same configuration space, and that a revival of configuration-space distinguishability rules out every Markovian realizer on that space.
\end{abstract}


\section{Introduction}

Quantum mechanics is a probabilistic theory that differs from classical stochastic theories in a precise way: while stochastic dynamics evolve probabilities directly, closed-system quantum evolution is governed by unitary transformations acting on complex amplitudes, with configuration probabilities obtained by taking their squared moduli.

It is precisely this phase-sensitive structure that gives rise to the distinctive quantum features of superposition, interference, and entanglement, which are traditionally regarded as setting quantum mechanics apart from classical stochastic processes. For a long time, this distinction was viewed as marking a fundamental divide between the two frameworks. 

More recently, however, a proposal known as the \textit{stochastic–quantum correspondence} (see Barandes’ original formulation and subsequent developments \cite{Barandes2025,Barandes2026Theorem,Barandes2024CausalLocal,Barandes2025Indivisible, BarandesHasanKagan2025, Pimenta2025Divisible,Calvo2026On}) has challenged this view by suggesting a striking alternative: that quantum theory may itself be understood as a particular class of stochastic processes, rather than as a fundamentally different dynamical framework.

The stochastic–quantum correspondence is formulated within a highly general setting, assuming only a configuration space together with probability distributions at given times and conditional probabilities governing their evolution. 

An important aspect emphasized in the original formulation is that no Markovian assumption is imposed, in contrast to what is commonly assumed in stochastic modeling (and, more broadly, physics in general). As a result, the dynamics may instead be \textit{indivisible}, in the sense that the evolution cannot be decomposed into a sequence of independent local transition steps. For this reason, the novel behavior captured by the correspondence arises not from the introduction of additional constraints, but rather from the relaxation of restrictive assumptions typically imposed on stochastic processes.

Nevertheless, the original construction makes a particular choice of lift from stochastic kernels to Hilbert-space dynamics, and this choice limits the class of quantum dynamics obtained from the correspondence. One aim of the present work is to make this structure explicit and to generalize the correspondence to general quantum channels.

Importantly, this limitation of the original exposition of the stochastic--quantum correspondence does not diminish its central achievement. In its textbook formulation, quantum theory is presented axiomatically, with measurement and the role of an observer placed at the center of the formalism. In its most familiar form, the Copenhagen interpretation, the observer plays a fundamental role not only in which observable is measured, but also in the post-measurement state-update (collapse) used to make future predictions.

Numerous approaches have sought to diminish the central role of observers in the dynamical update of quantum theory. These efforts often attempt to derive the Born rule and the collapse postulate from unitary dynamics alone, or by augmenting the standard axioms with additional assumptions.

In a stochastic description, these notions separate cleanly. A stochastic process comes with a fixed underlying configuration space. As time advances, a definite configuration is realized at each time, extending the realized trajectory. Each such realization is a physical fact, irrespective of whether it is observed. The probabilistic description of the future is updated by conditioning on this growing history, without requiring an observer-dependent measurement postulate.

While the original work is comprehensive, it leaves several natural questions. First, what restrictions does the original lifting prescription place on the reachable density matrices, and can it be generalized to encompass arbitrary quantum channels on the Hilbert space associated with the same configuration space? Second, does every rooted quantum dynamical family admit a stochastic realizer on the same configuration space, and what can rule out a Markovian realizer?

The aim of this work is to answer these questions and thereby clarify the relationship between stochastic processes and quantum dynamics.

The paper is organized as follows. Section~2 introduces stochastic processes and rooted dynamical families, distinguishing Markovianity, the CK property, and divisibility. Section~3 reviews the original lifting prescription and its restriction on reachable density matrices, then develops the generalized compatibility condition and CPTP dictionary. It shows that every stochastic kernel admits a compatible CPTP lift and that these lifts collectively encompass all quantum channels as the kernels vary. It also distinguishes compatibility from composable lifted evolution and examines the associated time-local Lindblad generators. Section~4 establishes the existence of stochastic realizers of rooted quantum dynamical families on the same configuration space, discusses their nonuniqueness, and gives a distinguishability-revival obstruction to Markovian realizers. Additional technical proofs and counterexamples are collected in the appendices.

\section{Stochastic processes, transition kernels, and dynamical families}
\label{sec:stochasticprocesses}
The theory of stochastic processes arises in a wide range of applied fields, including physics, finance and biology, most prominently through models based on Markov processes such as Brownian motion or continuous-time Markov chains. 

Over time, it has also been developed rigorously as a branch of probability theory. A major milestone in this development was the axiomatization of probability by Kolmogorov in 1933, which placed the theory on a firm measure-theoretic foundation. Because stochastic processes often involve conditioning, limits, or infinite-dimensional structures, their fully rigorous formulation can be technically involved. In the interest of clarity, we present only the relevant definitions while adopting a number of simplifying conventions that suffice for the present purposes.

\subsection{Stochastic processes}
A classical stochastic process indexed by time \(t \in T\) is a collection of random variables taking values in a configuration space \(X\),
\begin{equation}
	\{X_t : \Omega \rightarrow X\}_{t \in T},
\end{equation}
where \((\Omega,\Sigma,P)\) is a probability space. Here, \(\Omega\) denotes the sample space (which may be identified with a space of paths \(X^T\)), \(\Sigma\) is the \(\sigma\)-algebra of events, and \(P\) is a probability measure. In what follows we restrict to finite configuration spaces, \(X=\{1,\ldots,N\}\). 

By the Kolmogorov extension theorem (under suitable consistency conditions ensuring the existence of an underlying probability measure on path space), a stochastic process is equivalently characterized by its finite-time joint probability distributions.

For any finite collection of times \(t_0, \ldots, t_n \in T\), the joint
distribution
\begin{equation}
	p(x_n,\ldots,x_0) := 	P \bigl(X_{t_n}=x_n,\ldots,X_{t_0}=x_0\bigr),	\qquad x_i \in X,
\end{equation}
assigns probabilities to configurations of the process at those times.

Knowledge of all finite-time joint probability distributions $p(x_n, \ldots,x_0)$ completely specifies the stochastic process, in the sense that it fixes all probabilistic predictions associated with the process.

By the chain rule of probability, the joint distribution may be expanded as
\begin{equation}
	p(x_n,\ldots,x_0)	=	p(x_n \mid x_{n-1},\ldots,x_0)	p(x_{n-1} \mid x_{n-2},\ldots,x_0)	\cdots	p(x_1 \mid x_0)	p(x_0). \label{eq:chainrule}
\end{equation}

In this form, the dynamical structure of the stochastic process becomes explicit and admits a clear interpretation. The probability of a trajectory in path space up to time $t_n$ is updated via conditional probabilities $p(x_n \mid x_{n-1},\ldots,x_0)$ that may, in general, depend on the entire prior history of the path.

A stochastic process is \emph{Markov} if for any finite collection of times \(t_0 < \cdots < t_{n-1} < t_n\):
\begin{equation}
	p(x_n \mid x_{n-1},\ldots,x_0) = p(x_n \mid x_{n-1}),
	\label{eqn:markovprop}
\end{equation}

and hence the dynamical structure of the process becomes:
\begin{equation}
	p(x_n,\ldots,x_0)	=	p(x_n \mid x_{n-1})	p(x_{n-1} \mid x_{n-2})\cdots	p(x_1 \mid x_0)p(x_0).
\end{equation}

Therefore, in a Markov process, all finite-time joint distributions are determined by the initial distribution \(p(x_0)\) together with the family of two-time transition kernels \(p(x_t \mid x_s)\).

Although, in general, a collection of two-time conditional probabilities does not determine a unique stochastic process, these quantities remain fundamental. In particular, for any fixed global process the pairwise conditional probabilities \(p(x_t\mid x_s)\) are well defined whenever \(p(x_s)>0\). They determine the correct one-time marginal at time \(t\) from the global marginal at time \(s\) through
\begin{equation}
p(t)=\G{t}{s}p(s),
\qquad
\G{t}{s}:=\bigl[p(x_t\mid x_s)\bigr].
\label{eqn:twotimetransitionkernel}
\end{equation}

Here the square brackets denote the matrix of conditional probabilities, with rows indexed by configurations at time \(t\) and columns indexed by configurations at time \(s\). Explicitly, \(\G{t}{s}_{ji}=P(X_t=j\mid X_s=i)\).

\subsection{Dynamical families}
\label{sec:dynamicalmaps}
\label{par:stochastic-kernel-properties}

\paragraph{Classical dynamical families.}
For a fixed root time \(s\in T\), we define a \emph{classical dynamical family} rooted at \(s\) on \(X=\{1,\ldots,N\}\) to be a family of matrices 
\begin{equation}
	\bigl\{\Gamma(t\leftarrow s)\bigr\}_{t\geq s}
\end{equation}
satisfying the following conditions:

\begin{enumerate}
	\item Positivity.
	\begin{equation}
		\Gamma_{ji}(t\leftarrow s)\geq 0
		\qquad
		\forall t\geq s,\quad \forall i,j.
		\label{eqn:gammapos}
	\end{equation}
	
	\item Column normalization.
	\begin{equation}
		\sum_{j=1}^{N}\Gamma_{ji}(t\leftarrow s)=1
		\qquad
		\forall t\geq s,\quad \forall i.
		\label{eqn:colsumcomp}
	\end{equation}
	
	\item Identity at equal times.
	\begin{equation}
		\Gamma(s\leftarrow s)=I.
		\label{eq:identityprop}
	\end{equation}
\end{enumerate}

The first two conditions make each matrix column-stochastic; the third fixes the value of the family at its root.

The classical dynamical family determines \(p(t)=\Gamma(t\leftarrow s)p(s)\), but specifies no map between intermediate times; a single stochastic map taking \(p(u)\) to \(p(t)\) for every initial distribution \(p(s)\) need not exist.\footnote{The classical dynamical family defines a linear probability dynamics in the sense of Ref.~\cite{Egri2026}.} 

Any stochastic process on \(X\) induces such a dynamical family at any chosen root time \(s\) through its two-time conditional probabilities, as in Eq.~\eqref{eqn:twotimetransitionkernel}.\footnote{If \(P(X_s=i)=0\), the corresponding column may be chosen arbitrarily subject to positivity and normalization, with \(\Gamma_{ji}(s\leftarrow s)=\delta_{ji}\). These choices do not affect \(p(t)=\Gamma(t\leftarrow s)p(s)\). } Even together with \(p(s)\), the dynamical family does not generally determine higher-order joint distributions, and hence does not specify a unique stochastic process.

\paragraph{Unistochastic matrices and dynamical families.}
A matrix \(M\) is called \emph{unistochastic} \cite{Horn1954,Thompson1989,NylenTamUhlig1993} if there exists a unitary matrix \(U\) of the same dimension such that 
\begin{equation}
	M_{ji}=|U_{ji}|^2 \quad \forall i, j.
\end{equation}

We call a classical dynamical family \emph{unistochastic} if each matrix \(\Gamma(t\leftarrow s)\) is unistochastic. Within his framework of indivisible stochastic processes, Barandes~\cite{Barandes2025} calls a process \emph{unistochastic} when its specified transition matrices are unistochastic. His presentation takes a distinguished initial time as the conditioning time and describes the subsequent dynamics through stochastic kernels rooted at that time. In the terminology introduced here, these transition data constitute a unistochastic classical dynamical family, with the initial time serving as its root.

Barandes leaves higher-order conditional probabilities unspecified, interpreting the transition data as describing an equivalence class of compatible stochastic processes. We call these data a \emph{dynamical family} to distinguish them from a fully specified path-space law.

\paragraph{Quantum dynamical families.}
For later comparison with quantum dynamics, we introduce the quantum analogue of a classical dynamical family. For a fixed root time \(s\), a \emph{quantum dynamical family} rooted at \(s\) is a family
\begin{equation}
	\bigl\{\mathcal E_{t,s}\bigr\}_{t\geq s},
	\qquad
	\mathcal E_{s,s}=\mathrm{id},
\end{equation}
in which each \(\mathcal E_{t,s}\) is a CPTP map. Such a family is commonly called a \emph{quantum dynamical map} in the literature \cite{Benatti2024,Chruscinski2022}; we use ``family'' to distinguish it from its individual constituent maps. As in the classical case, this definition does not impose any composition law between intermediate times. Intermediate linear maps \(\tilde{\mathcal E}_{t,u}\) satisfying \(\mathcal E_{t,s}=\tilde{\mathcal E}_{t,u}\circ\mathcal E_{u,s}\), when they exist, are referred to as \emph{intertwiners} in the literature \cite{Benatti2024}.

\paragraph{Chapman--Kolmogorov property.}

This is an algebraic composition law for a supplied two-parameter family of maps. We abbreviate it as the \emph{CK property} and distinguish it from the \emph{Markov property}, which concerns conditional independence in the multipoint distributions of a stochastic process.

Fix a time interval \(I=[s,t_f]\), and let \(\{T_{t,u}\}\) be a two-parameter family of transformations defined for all \(u,t\in I\) with \(u\leq t\). We say that this family satisfies the \emph{CK property on \(I\)} if
\begin{equation}
	T_{t,u}
	=
	T_{t,v}\circ T_{v,u},
	\label{eq:ck_maps}
\end{equation}
for every choice of \(u,v,t\in I\) with \(u\leq v\leq t\), together with
\begin{equation}
	T_{u,u}=\mathrm{id},
\end{equation}
for all \(u\in I\). We also refer to a family satisfying the CK property as a \emph{composable} family.

This definition applies equally to full two-parameter families of classical transition kernels, \(T_{t,u}=\Gamma(t\leftarrow u)\), and of quantum channels, \(T_{t,u}=\mathcal E_{t,u}\). However, unlike the dynamical family defined above, for which the root time \(s\) is fixed, a full two-parameter family contains maps between every pair of ordered times. 

If the process is Markov in the sense of Eq.~\eqref{eqn:markovprop}, then its induced kernels satisfy the CK property:
\begin{equation}
	\G{t}{s}=\G{t}{u} \G{u}{s},
	\qquad \forall  s<u<t.
	\label{eq:ck}
\end{equation}

The converse does not in general hold: the two-time conditional kernels can satisfy the CK property even when the process violates the Markov conditional-independence property. The XOR process in Appendix~\ref{app:xor-process} provides an example.

\paragraph{Division events.}
Following Barandes~\cite{Barandes2025}, a classical dynamical family rooted at \(s\) is said to be \emph{divisible at \(u>s\)} if there is a family of column-stochastic matrices \(\{\widetilde\Gamma(t\leftarrow u)\}_{t\geq u}\), with \(\widetilde\Gamma(u\leftarrow u)=I\), such that
\begin{equation}
	\Gamma(t\leftarrow s)
	=
	\widetilde\Gamma(t\leftarrow u)\Gamma(u\leftarrow s),
	\qquad t\geq u.
	\label{eq:cdivision}
\end{equation}

The intermediate time \(u\) is called a \emph{division time}, and the system is said to undergo a \emph{division event} there. If no such family of stochastic dividers exists, the dynamics are said to be \emph{indivisible at \(u\)}.

This factorization does not by itself identify the divider with the actual two-time conditional kernel of a stochastic process; see Appendix~\ref{app:divider-not-kernel}.

\section{Quantum lifts}

\subsection{The stochastic--quantum correspondence}
\label{sec:stocasticquantum}
The stochastic--quantum correspondence begins from the elementary observation that the entries of any transition kernel are non-negative real numbers. Hence, for any matrix \(\Gamma(t \leftarrow s)\) describing the conditional probabilities for transitions from time \(s\) to time \(t\), one may always introduce quantities \(\theta_{ji}(t,s)\in\mathbb C\) satisfying
\begin{equation}
	\Gamma_{ji}(t \leftarrow s)=\lvert \theta_{ji}(t, s)\rvert^{2},
	\qquad \forall i,j\in X .
	\label{eq:gthetacomp}
\end{equation}
At this stage, equation \eqref{eq:gthetacomp} is simply a mathematical identity: each non-negative transition probability is represented as the modulus square of a complex number. The representation is entrywise, and does not involve matrix multiplication. This particular entrywise modulus-square representation is the starting point of the Barandes construction; later we will consider how changing this choice changes the class of quantum maps obtained from the same stochastic data.

Now regard \(\theta(t,s)\) as a complex matrix acting on \(\mathcal H\simeq \mathbb C^N\), equipped with the distinguished configuration basis \(\{e_i\}_{i\in X}\). Let
\[
\hat P^i := e_i e_i^\dagger
\]
denote the corresponding rank-one projectors. Then equation \eqref{eq:gthetacomp} can be written as
\begin{equation}
	\Gamma_{ji}(t \leftarrow s)
	=
	\operatorname{Tr} \left[
	\theta(t,s)^\dagger \hat P^j \theta(t,s)\hat P^i
	\right].
	\label{eq:barandes-dictionary}
\end{equation}
This is the so-called \emph{dictionary} of the stochastic--quantum correspondence. It encodes the transition probabilities of the stochastic process on the left-hand side in terms of Hilbert-space quantities involving the complex operator \(\theta(t,s)\) on the right-hand side.

Since $\Gamma$ is positive by construction, to generate a valid stochastic kernel, the operator \(\theta(t,s)\) need only satisfy the column-stochasticity condition \eqref{eqn:colsumcomp} and the identity condition \eqref{eq:identityprop}, which become, respectively,
\begin{align}
	\sum_{j\in X}\lvert \theta_{ji}(t,s)\rvert^{2}&=1,
	\qquad \forall i\in X,
	\label{eq:theta-column-stochasticity}\\
	\theta_{ji}(s,s)&=e^{i\phi_i(s)} \delta_{ji},
	\label{eq:theta-identity-condition}
\end{align}
where \(\delta_{ji}\) is the Kronecker delta and \(\phi_i(s)\) is an arbitrary real-valued phase.\footnote{The stronger condition \(\theta(s,s)=I\) is a choice of phase convention, rather than a consequence of the stochastic identity condition itself.}

Furthermore, the probability distribution at time \(s\) may be embedded into the Hilbert-space formalism by placing its entries on the diagonal of a matrix:
\begin{equation}
	\rho(s):=\sum_i \hat P^i  p(x_s=i).
	\label{eq:initial-diagonal-density}
\end{equation}
Given \(\theta(t,s)\), one may then define the associated \(\theta\)-conjugation map
\begin{equation}
	\phi_\theta(\rho)
	:=
	\theta(t,s)\rho\theta(t,s)^\dagger,
	\label{eq:theta_conj}
\end{equation}
and use it to define a time-dependent continuation of the initial diagonal state:
\begin{equation}
	\rho(t) := \phi_\theta(\rho(s)).
	\label{eq:rho-theta-evolution}
\end{equation}
By the identity condition, this continuation has the correct value at the root: \(\rho(t)\big|_{t=s}=\rho(s)\). Moreover, using the dictionary \eqref{eq:barandes-dictionary}, the law of total probability \eqref{eqn:twotimetransitionkernel} becomes
\begin{equation}
	p(x_t=j)
	=
	\operatorname{Tr} \left[
	\hat P^j \rho(t)
	\right].
	\label{eq:totalprob}
\end{equation}
This equation has the same form as the Born rule for a projective measurement onto the \(j\)th state of the configuration basis, with \(\rho(t)\) playing the role of the density matrix of a quantum state.

The matrix \(\rho(t)\) is a valid density matrix. Indeed, \(\rho(t)=\theta(t,s)\rho(s)\theta(t,s)^\dagger\) is Hermitian and positive semidefinite because \(\rho(s)\) is positive semidefinite, while
\[
\operatorname{Tr}\rho(t)
=\sum_i p(x_s=i)\sum_j|\theta_{ji}(t,s)|^2
=\sum_i p(x_s=i)
=1,
\]
by the column-stochasticity condition \eqref{eq:theta-column-stochasticity}.

The map \(\phi_\theta\) is completely positive for any \(\theta(t,s)\), since it has a Kraus representation with a single Kraus operator \(K=\theta(t,s)\). What is not automatic is trace preservation: \(\phi_\theta\) is not generally trace-preserving on arbitrary density matrices. In fact, it is trace-preserving and hence CPTP if and only if \(\theta(t,s)\) is unitary.\footnote{Suppressing the time arguments, if \(\theta\) is unitary then \(\theta^\dagger\theta=I\), so \(\phi_\theta\) is trace-preserving and admits a Kraus representation with a single Kraus operator \(\theta\), hence is CPTP. Conversely, if \(\phi_\theta(\rho)=\theta\rho\theta^\dagger\) is trace-preserving for all	density matrices \(\rho\), then	
	\[
	\operatorname{Tr}(\phi_\theta(\rho))
	=
	\operatorname{Tr}(\rho\theta^\dagger\theta)
	=
	\operatorname{Tr}(\rho),
	\]
	which implies \(\theta^\dagger\theta=I\). Since \(\theta\) is square, this is equivalent to unitarity.}

This condition identifies the special role of \emph{unistochastic} classical dynamical families. For such a family, each kernel admits a unitary representative \(U(t,s)\) satisfying
\begin{equation}
	\Gamma_{ji}(t\leftarrow s)=|U_{ji}(t,s)|^2 .
	\label{eq:unistochastic}
\end{equation}
Choosing \(\theta(t,s)=U(t,s)\) then gives the unitary quantum channel
\begin{equation}
	\phi_U(\rho)=U(t,s)\rho U(t,s)^\dagger .
	\label{eq:unitarychannel}
\end{equation}
Thus, whenever the classical dynamical family is unistochastic, its compatible lift may be chosen to consist of unitary quantum channels. The stochastic--quantum correspondence therefore includes unitary evolution, the characteristic form of closed-system quantum dynamics.

By contrast, for a general stochastic kernel, the one-Kraus map \(\phi_\theta\) need not be trace-preserving and therefore need not define a quantum channel.

The key move in the Barandes construction is the replacement, on diagonal inputs, of the \(\theta\)-conjugation formula \eqref{eq:rho-theta-evolution} by a CPTP map. Following Barandes, this
replacement is implemented by introducing the column-selector Kraus operators
\begin{equation}
	K_\beta(t,s):=\theta(t,s)\hat P^\beta .
	\label{eq:column-selector-kraus}
\end{equation}
These operators are so named because \(K_\beta(t,s)\) retains only the \(\beta\)-th column of \(\theta(t,s)\) and sends all other columns to zero.

Since \(\rho(s)\) is diagonal in the distinguished configuration basis, the column-selector construction reproduces the same output as \(\theta\)-conjugation:
\begin{align}
	\sum_\beta K_\beta(t,s)\rho(s)K_\beta(t,s)^\dagger
	&=
	\sum_\beta
	\theta(t,s)\hat P^\beta \rho(s)\hat P^\beta
	\theta(t,s)^\dagger  \\
	&=
	\sum_\beta
	p(x_s=\beta)\theta(t,s)\hat P^\beta\theta(t,s)^\dagger \\
	&=
	\theta(t,s)\rho(s)\theta(t,s)^\dagger
	=
	\rho(t).
	\label{eq:column-selector-map}
\end{align}
Thus, on diagonal inputs, the Kraus representation with operators \(K_\beta(t,s)\) gives the same density matrix \(\rho(t)\) as the original \(\theta\)-conjugation formula.

The summation condition \eqref{eq:theta-column-stochasticity} ensures the Kraus completeness relation
\begin{equation}
	\sum_\beta K_\beta(t,s)^\dagger K_\beta(t,s)
	=
	\sum_\beta\left(\sum_j|\theta_{j\beta}(t,s)|^2\right)
	\hat P^\beta
	=I.
	\label{eqn:colselkraus}
\end{equation}

Hence the column-selector construction defines a CPTP map,
\begin{equation}
	\Phi_\theta(\rho)
	:=
	\sum_\beta K_\beta(t,s)\rho K_\beta(t,s)^\dagger,
	\label{eq:column-selector-channel}
\end{equation}
and therefore a quantum channel.

In particular, the two maps \(\phi_\theta\) and \(\Phi_\theta\) induce the same stochastic kernel on configuration basis states:
\begin{equation}
	\Gamma_{ji}(t\leftarrow s)
	=
	\operatorname{Tr} \left[
	\hat P^{j}\phi_\theta(\hat P^{i})
	\right]
	=
	\operatorname{Tr} \left[
	\hat P^{j}\Phi_\theta(\hat P^{i})
	\right].
\end{equation}
Thus the Kraus map \(\Phi_\theta\) agrees with the conjugation map \(\phi_\theta\) on diagonal inputs, but generally not on arbitrary density matrices.

In the language of quantum channels, the Kraus map \eqref{eq:column-selector-channel} may be understood as a non-selective projective measurement, or L\"uders dephasing, in the configuration basis,
followed by conjugation with \(\theta(t,s)\).

Strictly following the Barandes construction, one is therefore led to two natural quantum channels associated with a stochastic kernel. In the unistochastic case, a unitary representative \(\theta(t,s)=U(t,s)\) gives the unitary channel \eqref{eq:unitarychannel}. For a general stochastic kernel, the column-selector construction gives the CPTP map \(\Phi_\theta\). These channels are sufficient to encode any stochastic transition kernel, and hence the corresponding two-time law of total probability, in terms of quantum channels on a Hilbert space.

Nevertheless, the construction does not span the full class of CPTP maps on \(\mathcal B(\mathcal H)\). Instead, it produces only restricted families—unitary channels in the unistochastic case and, more generally, channels arising from the column-selector Kraus construction. For a fixed diagonal initial state, not every target state \(\rho(t)\) is reachable within this framework, as we now show. 

We can categorize the reachable states as follows. Since \(\phi_\theta\) and \(\Phi_\theta\) agree on diagonal inputs, the initial state \eqref{eq:initial-diagonal-density} is mapped, in either case, to
\begin{align}
	\sigma
	&=
	\sum_i p(x_s=i) \theta(t,s)\hat P^i\theta(t,s)^\dagger\\
	&=
	\sum_i p(x_s=i) |\psi_i\rangle\langle \psi_i|.
\end{align}
Here \(|\psi_i\rangle := \theta(t,s)e_i\) has unit norm by the column-normalization condition~\eqref{eq:theta-column-stochasticity}.

Thus \(\sigma\) is reachable only if it admits a pure-state ensemble decomposition with fixed weights \(p(x_s)\). By the ensemble-majorization theorem~\cite{Nielsen1999}, the existence of such an ensemble is equivalent to
\[
p(x_s)\prec \lambda(\sigma),
\]
where \(p(x_s)\) is the initial probability vector and \(\lambda(\sigma)\) is the eigenvalue vector of \(\sigma\), both arranged in non-increasing order. In words, the reachable states are constrained to be density matrices whose spectrum majorizes \(p(x_s)\).

Two standard consequences of majorization are that the Shannon entropy \(H\) is Schur-concave and that the quadratic function \(\sum_i x_i^2\) is Schur-convex. Therefore, the condition \(p(x_s)\prec\lambda(\sigma)\) implies
\begin{equation}
H(p(x_s))\geq H(\lambda(\sigma)),
\qquad
\sum_i p(x_s=i)^2\leq \sum_i \lambda_i(\sigma)^2 .
\end{equation}
Since the von Neumann entropy satisfies \(S(\sigma)=H(\lambda(\sigma))\), while \(\rho=\operatorname{diag}(p(x_s))\) implies \(S(\rho)=H(p(x_s))\) and \(\operatorname{Tr}(\rho^2)=\sum_i p(x_s=i)^2,
\) it follows that
\begin{equation}
S(\sigma)\leq S(\rho),
\end{equation}
and
\begin{equation}
\operatorname{Tr}(\sigma^2)\geq \operatorname{Tr}(\rho^2).
\end{equation}
Thus, although the constructed maps can change the diagonal probabilities of the output, they cannot map a diagonal state to a genuinely more mixed density matrix. Relative to the initial diagonal state, they can only preserve or increase the purity.

The entropy bound also has a direct interpretation through Jaynes's minimum-entropy characterization of the spectral decomposition of a density matrix \cite[Appendix A]{Jaynes1957}. The construction preserves the initial ensemble weights while replacing the orthogonal basis states with normalized, potentially nonorthogonal pure states. Since the von Neumann entropy is the minimum Shannon entropy of the ensemble weights over all pure-state ensemble decompositions, it cannot exceed the entropy of these preserved weights, which equals the initial von Neumann entropy.\footnote{We thank Jacob Barandes for pointing out this interpretation in terms of Jaynes's minimum-entropy characterization (private communication).}

At first sight, the resulting restriction on the lifted quantum dynamics -- specifically, the inability of the construction to reach arbitrary density matrices from some diagonal initial states -- appears difficult to reconcile with the broader aim of the stochastic--quantum correspondence: to recover the full structure of quantum mechanics from indivisible stochastic processes.

The distinction to keep in mind is that the restriction above concerns the quantum dynamics obtained directly on the Hilbert space associated with the given configuration space.

We note that Barandes also appeals to the Stinespring dilation theorem~\cite{Stinespring1955} to give unitary dilations of the above quantum-channel constructions. In particular, the column-selector map \(\Phi_\theta\) can be represented as unitary evolution on a larger Hilbert space followed by tracing out an ancillary system. This provides a unitary realization of \(\Phi_\theta\) at the dilated level, but it does not enlarge the class of reduced dynamics obtained from the original stochastic kernel: after tracing out the ancilla, one recovers the same channel \(\Phi_\theta\), not an arbitrary CPTP map.

This does not, however, preclude obtaining arbitrary CPTP maps by choosing a suitable enlarged unistochastic starter system. Given an arbitrary CPTP map on the system of interest, one may first construct a unitary dilation and then use its associated unistochastic transition kernel on the enlarged composite configuration space as the starting point of the correspondence. Using the dilating unitary for the lift, one then recovers the desired CPTP map as reduced dynamics on the system of interest.

Thus arbitrary CPTP dynamics can be accommodated by passing to a larger starter system. What remains open is whether this enlargement is necessary, or whether the restriction arises from the particular Hilbert-space representation used in the construction. This motivates the following question: without enlarging the configuration space or the associated Hilbert space, can the same class of stochastic transition kernels admit a more general lifting prescription that encompasses all CPTP maps? As we show in the next subsection, the answer is affirmative. 

To this end, we do not assume the representation \eqref{eq:gthetacomp} from the outset. Instead, we study possible Hilbert-space lifts of a stochastic kernel in a more general operator-theoretic setting, recovering \eqref{eq:gthetacomp} as one distinguished special case.

This more abstract formulation makes explicit the assumptions linking stochastic kernels to their operator lifts.

\subsection{Generalizing the correspondence}
\label{sec:gencor}

While the original formulation uses the \(\theta\)-representation \eqref{eq:gthetacomp}, it may be viewed more broadly as constructing a linear map on \(\mathcal B(\mathcal H)\) whose diagonal-to-diagonal restriction reproduces the prescribed stochastic kernel.

In what follows, we systematically investigate linear maps with this property.

\subsubsection{Operator lifts and the compatibility condition}

Fix a configuration basis \(\{|i\rangle\}_{i=1}^N\) of an \(N\)-dimensional Hilbert space \(\mathcal H\). We represent a probability distribution at time \(s\) by a column vector \(p(s)\in\mathbb R^N\), with \(p_i(s)\ge 0\) and \(\sum_i p_i(s)=1\), and a two-time kernel \(\G{t}{s}\) by a linear map on \(\mathbb R^N\), represented in this basis by a column-stochastic matrix.

We seek a linear map
\[
\phi:\mathcal B(\mathcal H)\to \mathcal B(\mathcal H)
\]
such that, for diagonal inputs, the diagonal part of \(\phi\) reproduces the stochastic evolution. Note that, in general, the range of \(\phi\) may have off-diagonal components.

Define the diagonal embedding \(J:\mathbb R^N \to \mathcal B(\mathcal H)\) by

\begin{equation}
J(p)=\mathrm{diag}(p_1,p_2,\ldots,p_N)
=\sum_{i=1}^N p_i\hat P^{i},
\end{equation}
which sends probability vectors to diagonal operators in the configuration basis. Further let \(\Pi:\mathcal B(\mathcal H)\to \mathcal B(\mathcal H)\) denote the projection onto the diagonal in the configuration basis,
\begin{equation}
	\Pi(\rho):=\sum_{j=1}^N \hat P^{j}\rho\hat P^{j}.
\end{equation}

Then the desired compatibility condition between the map $\phi_{t,s}$ and the stochastic kernel $\G{t}{s}$ is
\setlength{\fboxsep}{8pt}
\begin{empheq}[box=\fbox]{equation}
	\Pi \bigl(\phi_{t,s}(J(p(s)))\bigr)
	=
	J \bigl(\G{t}{s} p(s)\bigr),\strut
	\qquad\text{\textsc{(Compatibility)}}
	\label{eqn:compat}
\end{empheq}
for all probability column vectors \(p(s)\). The compatibility condition is the commutative-diagram requirement that characterizes the lift; see Fig.~\ref{fig:compatibility-diagram}. In words, evolving the embedded probability distribution by \(\phi_{t,s}\) and then reading out its diagonal part gives the same probability distribution as evolving directly by the stochastic kernel \(\G{t}{s}\).

\begin{figure}[t]
	\centering
	{\large$
		\begin{tikzcd}[column sep=huge, row sep=large, scale=1.15, transform shape]
			J(p(t)) \arrow[d, "D"]  & J(p(s))  \arrow[l, "{\Pi \circ \phi_{t,s}}"] \\
			p(t) & p(s) \arrow[l, "{ \Gamma(t\leftarrow s)}"] \arrow[u, "J"']
		\end{tikzcd}
		$}
	\caption{Commuting diagram for the lift compatibility condition. The map \(D\) denotes the readout map from a diagonal operator to the corresponding probability vector, with \(D\circ J=\mathrm{id}\). The compatibility condition is \(D\Pi\phi_{t,s}J=\Gamma(t\leftarrow s)\). In general, \(\phi_{t,s}\) may map a diagonal operator to one with off-diagonal components; the projection \(\Pi\) extracts the diagonal sector, thereby providing the link back to ordinary probabilities.	}
	\label{fig:compatibility-diagram}
\end{figure}

Taking \(p(s)\) to be the probability vector concentrated on configuration \(i\), so that \(J(p(s))=\hat P^i\), and using \([\Pi(A)]_{jj}=\operatorname{Tr}(\hat P^j A)\), the \(j\)th diagonal entry of Eq.~\eqref{eqn:compat} gives
\begin{equation}
	\G{t}{s}_{ji}
	=
	\operatorname{Tr}\bigl(\hat P^{j}\phi_{t,s}(\hat P^{i})\bigr).
	\label{eqn:gammatrace}
\end{equation}

This is the component form of the compatibility condition. Conversely, by linearity, Eq.~\eqref{eqn:gammatrace} implies Eq.~\eqref{eqn:compat} for arbitrary probability vectors \(p(s)\).

For fixed root time \(s\), the induced kernels \(\{\Gamma(t\leftarrow s)\}_{t\geq s}\) must form a rooted classical dynamical family in the sense of Section~\ref{sec:dynamicalmaps}. Through Eq.~\eqref{eqn:gammatrace}, the maps \(\phi_{t,s}\) must satisfy the following corresponding requirements for all $t\geq s$:

\begin{enumerate}
	\item \textbf{Positivity:} 
	\begin{equation}
		\operatorname{Tr}\bigl(\hat P^j\phi_{t,s}(\hat P^i)\bigr)\geq 0
		\qquad \forall  i,j.
		\label{eqn:phiposdiag}
	\end{equation}
	
	\item \textbf{Column normalization:} 
	\begin{equation}
		\operatorname{Tr}\bigl(\phi_{t,s}(\hat P^i)\bigr)=1
		\qquad \forall  i.
		\label{eqn:phitracediagbasis}
	\end{equation}
	
	\item \textbf{Identity at equal times:}
	\begin{equation}
		\operatorname{Tr}\bigl(\hat P^j\phi_{s,s}(\hat P^i)\bigr)
		=
		\delta_{ji}\qquad \forall i,j.
		\label{eqn:phiidentitydiag}
	\end{equation}
\end{enumerate}

Note these place conditions only on the diagonal-to-diagonal restriction of \(\phi_{t,s}\). They do not require \(\phi_{t,s}\) to be positive or trace preserving, or \(\phi_{s,s}\) to equal the identity map on the full operator algebra \(\mathcal B(\mathcal H)\).

A general left--right operator-sum representation of these linear lifts, together with the corresponding stochastic-kernel conditions, is given in Appendix~\ref{app:ABproof}.

In the special case \(\phi_{t,s}(\rho)=\theta(t,s)\rho\theta(t,s)^\dagger\), the compatibility condition \eqref{eqn:gammatrace} reduces to the Barandes dictionary \eqref{eq:barandes-dictionary}:

\begin{equation} 
	\G{t}{s}_{ji} = \operatorname{Tr}\bigl(\hat P^{j}\theta(t,s)\hat P^{i} \theta(t,s)^\dagger\bigr) = |\theta_{ji}(t,s)|^2.  
\end{equation}

Substituting the conjugation lift \(\phi_{t,s}(\rho)=\theta(t,s)\rho\theta(t,s)^\dagger\) into the stochastic-kernel requirements stated above yields the corresponding conditions on \(\theta(t,s)\), given in Sec.~\ref{sec:stocasticquantum}.

\subsubsection{Completely positive, trace-preserving lifts}
\label{sec:cptplifts}
To place the above construction in the setting of quantum dynamics, we now restrict to compatible lifts that are CPTP. Such maps are known as quantum channels.

This restriction is stronger than the first two diagonal stochastic-kernel requirements stated above, and is sufficient for them. Indeed, suppose \(\phi_{t,s}\) is CPTP. Since complete positivity implies positivity, we have \(\phi_{t,s}(\hat P^i)\geq 0\), which implies \eqref{eqn:phiposdiag}. Trace preservation gives
\begin{equation}
	\operatorname{Tr}\bigl(\phi_{t,s}(\hat P^i)\bigr)
	=
	\operatorname{Tr}(\hat P^i)
	=
	1,
\end{equation}
and hence the induced kernel is column normalized, as required by \eqref{eqn:phitracediagbasis}.

Complete positivity and trace preservation do not imply the equal-time condition \eqref{eqn:phiidentitydiag}, which must still be imposed independently. If we further require \(\phi_{s,s}=\mathrm{id}\) on \(\mathcal B(\mathcal H)\), then \(\{\phi_{t,s}\}_{t\geq s}\) forms a rooted quantum dynamical family as defined in Section~\ref{sec:dynamicalmaps}. This stronger initial condition is not required for compatibility with a rooted classical dynamical family.\footnote{For example, the column-selector channel \eqref{eq:column-selector-channel} has 	\(\Phi_{\theta(s,s)}=\Pi\), the dephasing channel, rather than the identity channel on \(\mathcal B(\mathcal H)\).}

Every completely positive map admits a Kraus representation~\cite{Kraus1971}
\begin{equation}
	\phi_{t,s}(\rho)
	=
	\sum_{\beta=1}^{r}
	K_\beta(t,s)\rho K_\beta(t,s)^\dagger .
\end{equation}

Moreover, the map is trace preserving if and only if the Kraus operators satisfy
\begin{equation}
	\sum_{\beta=1}^{r}
	K_\beta(t,s)^\dagger K_\beta(t,s)
	=
	I .
\end{equation}

In Kraus form, the identity-at-equal-times condition \eqref{eqn:phiidentitydiag} becomes
\begin{equation}
	\sum_{\beta=1}^{r}
	|K_{\beta,ji}(s,s)|^2
	=
	\delta_{ji}.
\end{equation}

Substituting this representation into Eq.~\eqref{eqn:gammatrace} gives the CPTP form of the dictionary
\begin{equation}
	\Gamma(t\leftarrow s)
	=
	\sum_{\beta=1}^{r} K_\beta(t,s)\odot K_\beta(t,s)^*,
	\label{eqn:Kdict}
\end{equation}
where \(\odot\) denotes the Schur--Hadamard product \((X\odot Y)_{ij}:=X_{ij}Y_{ij}\), in the configuration basis, and \(^{*}\) denotes entrywise complex conjugation. Equivalently, in components,
\begin{equation}
	\Gamma_{ji}(t\leftarrow s)
	=
	\sum_{\beta=1}^{r}
	\bigl|K_{\beta,ji}(t,s)\bigr|^2.
	\label{eqn:Kdictcomp}
\end{equation}

While Eq.~\eqref{eqn:Kdictcomp} appears in the original work of Barandes (see \cite{Barandes2025}), there it arises from the specific column-selector Kraus operators \eqref{eq:column-selector-kraus} discussed in Section~\ref{sec:stocasticquantum}. The generalized correspondence developed here permits arbitrary CPTP maps consistent with the dictionary \eqref{eqn:Kdict}. This answers the question posed in Section~\ref{sec:stocasticquantum}: the restriction that prevented the construction from accommodating arbitrary CPTP channels comes from the particular lifting prescription. We retain the same full class of two-time kernels used by Barandes and generalize their Hilbert-space representation, without enlarging the configuration space or its associated Hilbert space. As the stochastic kernel ranges over this class, its compatible CPTP lifts collectively encompass every CPTP channel on this Hilbert space.

This also clarifies the role of the original representation \eqref{eq:gthetacomp}. Instead of postulating at the outset that each transition probability is the modulus square of a single complex amplitude, we derive the more general form \eqref{eqn:Kdictcomp} from the structural requirements that the lift be compatible with the stochastic kernel on the diagonal-to-diagonal restriction and that it define a CPTP map.

In this generalized form, each transition probability is represented not by a single modulus square, but by a sum of modulus squares of complex amplitudes. This is the natural generalization of the original observation that transition probabilities are non-negative real numbers:
\begin{equation}
|\theta_{ji}(t,s)|^2
\quad\leadsto\quad
\sum_{\beta=1}^{r}
\bigl|K_{\beta,ji}(t,s)\bigr|^{2}.
\end{equation}
Thus, although one might initially wonder why the modulus-square representation should be preferred over some other positive function of complex amplitudes, we have shown that requiring a compatible CPTP lift of the stochastic kernel forces precisely the sum-of-modulus-squares form, of which the original modulus-square dictionary is the single-term algebraic case.

Finally, we note that every column-stochastic kernel admits at least one CPTP compatible lift. For example, define the Kraus operators
\begin{equation}
	K_{(j,i)}(t,s)
:=
\sqrt{\Gamma_{ji}(t\leftarrow s)} |j\rangle\langle i|,
\end{equation}
where the pair \((j,i)\) labels the operators. These define a completely positive map
\begin{equation}
\phi_{t,s}(\rho)
=
\sum_{i,j}K_{(j,i)}(t,s)\rho K_{(j,i)}(t,s)^\dagger .
\end{equation}
Moreover,
\begin{equation}
\sum_{i,j}K_{(j,i)}(t,s)^\dagger K_{(j,i)}(t,s)
=
\sum_i\left(\sum_j\Gamma_{ji}(t\leftarrow s)\right)
|i\rangle\langle i|
=
I,
\end{equation}
so the map is trace preserving. It also satisfies the compatibility condition
\eqref{eqn:gammatrace}:
\begin{align}
	\operatorname{Tr}\bigl(\hat P^j\phi_{t,s}(\hat P^i)\bigr)
	&=
	\sum_{i',j'}
	\operatorname{Tr}\bigl(
	\hat P^j K_{(j',i')}(t,s)
	\hat P^i K_{(j',i')}(t,s)^\dagger
	\bigr)\nonumber
	\\
	&=
	\Gamma_{ji}(t\leftarrow s).
\end{align}

\subsubsection{Composition and Lindblad generators}

We have shown that every rooted classical dynamical family admits a compatible family of CPTP lifts on \(\mathcal B(\mathcal H)\). Merely rewriting the prescribed stochastic kernels in this form does not supply new information about the stochastic evolution. 

One important reason for using the Hilbert-space representation is that it may provide a composable dynamical description where none exists at the level of the classical transition kernels. Consider a rooted classical family \(\{\G{t}{s}\}_t\) that is indivisible at some intermediate time \(u>s\). No family of column-stochastic dividers \(\{\widetilde\Gamma(t\leftarrow u)\}_{t\geq u}\) satisfies
\begin{equation}
	\G{t}{s}
	=
	\widetilde\Gamma(t\leftarrow u)\G{u}{s},
	\qquad t\geq u.
\end{equation}
The classical family may nevertheless admit a quantum dynamical family of CPTP maps \(\{\phi_{t,s}\}_t\), with intermediate CPTP
maps \(\phi_{t,u}\) satisfying:
\begin{equation}
	\phi_{t,s}
	=
	\phi_{t,u}\circ\phi_{u,s},
	\qquad s\leq u\leq t.
	\label{eqn:rooted-factorization}
\end{equation}

In such a completion, the lifted state at time \(u\) is sufficient to continue the evolution through the intertwiner \(\phi_{t,u}\), even though the intermediate classical distribution is not sufficient to do so through a stochastic kernel. 

One systematic way to construct a composable completion is through a sufficiently regular time-local generator. We begin by requiring the intermediate CPTP maps to satisfy the equal-time condition
\[
\phi_{u,u}=\mathrm{id}
\quad\text{on }\mathcal B(\mathcal H)
\]
at every time \(u\). At the prescribed root time \(s\), the condition \(\phi_{s,s}=\mathrm{id}\) makes the compatible CPTP lifts \(\{\phi_{t,s}\}_{t\geq s}\) a rooted quantum dynamical family in the sense of Section~\ref{sec:dynamicalmaps}. Requiring the same condition at every \(u\) further ensures that each family of intertwiners \(\{\phi_{t,u}\}_{t\geq u}\) begins with the identity at its own initial time.

Next, we assume that \(t\mapsto\phi_{t,s}\) is continuous and that the short-time intertwiner maps admit a time-local generator that depends continuously on \(t\):
\begin{equation}
	\mathcal L_t
	:=
	\lim_{h\downarrow0}
	\frac{\phi_{t+h,t}-\mathrm{id}}{h}.
	\label{eq:liftgenerator}
\end{equation}

For the fixed root time \(s\), Eq.~\eqref{eqn:rooted-factorization} gives
\[
\frac{\phi_{t+h,s}-\phi_{t,s}}{h}
=
\frac{\phi_{t+h,t}-\mathrm{id}}{h}
\circ\phi_{t,s},
\qquad h>0.
\]
Taking \(h\downarrow0\), one obtains\footnote{The limit initially gives a right derivative. Continuity of \(\mathcal L_t\circ\phi_{t,s}\) makes this the ordinary derivative for \(t>s\); at \(t=s\), it is understood from the right.}
\begin{equation}
	\partial_t\phi_{t,s}
	=
	\mathcal L_t\circ\phi_{t,s}.
	\label{eq:liftmastermap}
\end{equation}

Consequently, \(\rho(t)=\phi_{t,s}(\rho(s))\) satisfies
\begin{equation}
	\frac{\partial}{\partial t}\rho(t)
	=
	\mathcal L_t(\rho(t)).
	\label{eq:liftmasterstate}
\end{equation}

Under the assumptions above, this is a time-dependent Lindblad master equation, and \(\mathcal L_t\) is its time-dependent Lindblad generator. Specifically, \(\mathcal L_t\) has Gorini--Kossakowski--Sudarshan--Lindblad (GKSL) form \cite{GKS1976,Lindblad1976,Hall2014}:
\begin{equation}
	\mathcal L_t(\rho)
	=
	-i[H(t),\rho]
	+
	\sum_\mu
	\Bigl(
	L_\mu(t)\rho L_\mu(t)^\dagger
	-\tfrac12\{L_\mu(t)^\dagger L_\mu(t),\rho\}
	\Bigr),
	\label{eqn:lindblad}
\end{equation}
where \(H(t)=H(t)^\dagger\) and
\(L_\mu(t)\in\mathcal B(\mathcal H)\).

The time-local generator determines the full two-parameter family, including the rooted maps and their intertwiners, through
\begin{equation}
	\phi_{t,u}
	=
	\mathcal T\exp\left(
	\int_u^t \mathcal L_\tau d\tau
	\right),
	\qquad s\leq u\leq t.
	\label{eq:timeorderedlift}
\end{equation}
Here \(\mathcal T\) denotes time ordering, and setting \(u=s\) recovers the rooted maps \(\phi_{t,s}\). Compatibility with the prescribed rooted kernels then requires:
\begin{equation}
	\G{t}{s}_{ji}
	=
	\operatorname{Tr}\left[
	\hat P^j
	\left(
	\mathcal T\exp\left(
	\int_s^t\mathcal L_\tau d\tau
	\right)
	\right)(\hat P^i)
	\right]
	\label{eq:generatorcompatibility},
\end{equation}
for all \(t\geq s\) and all \(i,j\). Conversely, a continuous GKSL generator, $\mathcal L_t$, satisfying Eq.~\eqref{eq:generatorcompatibility} defines a compatible composable CPTP completion through Eq.~\eqref{eq:timeorderedlift}.

Once a  composable CPTP completion has been chosen, each short-time intertwiner admits a Kraus representation
\begin{equation}
\phi_{t+dt,t}(\rho)
=
\sum_\alpha
M_\alpha(t+dt,t)\rho M_\alpha(t+dt,t)^\dagger.
\end{equation}

For the GKSL generator \eqref{eqn:lindblad}, these Kraus operators can be chosen in the form
\begin{align}
	M_0(t+dt,t)
	&=
	I
	-
	dt\left(
	iH(t)+\frac12\sum_\mu L_\mu(t)^\dagger L_\mu(t)
	\right)
	+o(dt),
	\\
	M_\mu(t+dt,t)
	&=
	\sqrt{dt} L_\mu(t)+o(\sqrt{dt}).
\end{align}
Substitution into the Kraus representation gives
\begin{equation}
	\phi_{t+dt,t}(\rho)
	=
	\rho+dt \mathcal L_t(\rho)+o(dt).
\end{equation}
Thus in this short-time Kraus choice, the Lindblad operators \(L_\mu(t)\) appear as the \(\sqrt{dt}\)-coefficients in the short-time Kraus expansion of the intertwiners, while the Hamiltonian and dissipative drift terms are encoded in the first-order correction to \(M_0\).

Note that the role of the above completion is not to supply an intermediate lifted state, since
\[
\rho(u)=\phi_{u,s}(J(p(s)))
\]
is already defined by the rooted lift. Rather, the  composable completion ensures that this state is sufficient to continue the chosen lifted dynamics: for every \(t\geq u\),
\[
\rho(t)=\phi_{t,u}(\rho(u)).
\]
While the corresponding classical distribution can always be recovered from \(\rho(u)\) as
\[
p(u)
=
D\bigl(\Pi(\rho(u))\bigr)
=
\G{u}{s}p(s),
\]
it need not itself be sufficient to determine the subsequent evolution. The off-diagonal components of \(\rho(u)\) can therefore carry information relevant to future configuration probabilities that is absent from \(p(u)\).

When a compatible composable CPTP family with a generator exists, the prescribed rooted evolution consequently admits a step-by-step description through a Lindblad master equation, even when its rooted classical kernels are indivisible. 

\section{Stochastic realizers of quantum dynamical families}
\label{app:stochastic-representations}

Fix a configuration basis with projectors \(\hat P^i=|i\rangle\langle i|\), and a quantum dynamical family \(\{\phi_{t,s}\}_{t\geq s}\). Its restricted stochastic shadow in the configuration basis is the classical dynamical family  \(\{\G{t}{s}\}_{t\geq s}\), where
\begin{equation} 
	\G{t}{s}_{ji} 
	:= 
	\operatorname{Tr} \left[\hat P^j\phi_{t,s}(\hat P^i)\right]. 
	\label{eq:endpoint-rooted-shadow} 
\end{equation} 

For an initially diagonal state \(J(p_0):=\sum_i p_0(i)\hat P^i\), the configuration distribution at time \(t\) is \(\G{t}{s}p_0\).

In this section, we ask whether these configuration distributions can be reproduced by a fixed conditional path structure on the same configuration space, with only the initial distribution allowed to vary.

More precisely, let \(\mathbb P_\star\) be a reference path law with initial distribution \(p_\star(i):=\mathbb P_\star(X_s=i)>0\). Changing only the initial distribution to \(p_0\) defines
\begin{equation}
	\mathbb P_{p_0}(A)
	:=
	\sum_i p_0(i) \mathbb P_\star(A\mid X_s=i)
	\label{eq:endpoint-initial-reweighting}
\end{equation}
for every measurable path event \(A\). The conditional path laws on the right are fixed; only their initial weights vary. We ask whether the reference process can be chosen so that 
\begin{equation}
	\mathbb P_{p_0}(X_s=i,X_t=j)
	=p_0(i)\G{t}{s}_{ji}
	\qquad(t\geq s)
	\label{eq:endpoint-common-representation}
\end{equation}
for every \(p_0\). We call this fixed conditional path structure a \emph{stochastic realizer} of the rooted quantum dynamical family.\footnote{The terminology follows Barandes's use of \emph{non-Markovian realizer} for a completion of the prescribed dynamical data to a full hierarchy of conditional probabilities~\cite{Barandes2025Indivisible}. A realizer as defined here induces a stochastic process family 	in the sense of Egri et al.~\cite{Egri2026}, indexed by the initial	distribution at a fixed initial time. The conditional path laws given the initial configuration remain fixed as that distribution varies.}

We show that such a stochastic realizer always exists, but a Markovian one may not exist. The relation between quantum probability evolution and stochastic process descriptions, including the possibility of Markovian realizers, has been the subject of earlier debate \cite{Gillespie1994,GarbaczewskiOlkiewicz1996,Gillespie1996Reply} and recent analysis \cite{Egri2026}.

\subsection{Existence of a stochastic realizer}
\label{app:endpoint-process}

For any rooted classical dynamical family \(\{\G{t}{s}\}_t\) and initial distribution \(p_0\), define the finite-dimensional joint distributions at arbitrary times \(s=t_0<t_1<\cdots<t_m\) by
\begin{equation}
	p_{\mathrm{ep}}(x_0,\ldots,x_m)
	:=
	p_0(x_0)\prod_{k=1}^{m}\G{t_k}{s}_{x_kx_0}.
	\label{eq:endpoint-process-law}
\end{equation}
We call the resulting conditional path structure the \emph{endpoint realizer}. Its later configurations have been chosen to be conditionally independent given \(x_0\), so on histories of
positive probability
\begin{equation}
	p_{\mathrm{ep}}(x_k\mid x_{k-1},\ldots,x_0)
	=
	p_{\mathrm{ep}}(x_k\mid x_0).
	\label{eq:endpoint-process-conditional}
\end{equation}
Column normalization makes the joint laws, $p_{\mathrm{ep}}(x_0,\ldots,x_m)$, normalized and consistent under marginalization. They therefore define a stochastic process, including on a continuous time index set by the Kolmogorov extension theorem.

Summing over the other configurations gives the required matching conditions
\begin{equation}
	\Pr_{\mathrm{ep}}(X_s=i,X_t=j)
	=
	p_0(i)\G{t}{s}_{ji},
	\qquad
	\Pr_{\mathrm{ep}}(X_t=j)
	=
	\operatorname{Tr} \left[\hat P^j\phi_{t,s}(J(p_0))\right].
	\label{eq:endpoint-process-marginals}
\end{equation}

The conditional law in Eq.~\eqref{eq:endpoint-process-conditional} is independent of \(p_0\), so the same conditional path structure realizes the prescribed classical dynamical family for every initial distribution. The endpoint construction shows that every rooted classical dynamical family admits a stochastic realizer. However, the data specified by the classical dynamical family do not, in general, determine the higher-order joint probabilities uniquely. The endpoint realizer is one possible choice; the same rooted classical dynamical family may admit other stochastic realizers.

\subsection{Failure of Markov compatibility}
\label{sec:markov-compatible-realizers}
Although the same rooted classical dynamical family may admit different realizers, a Markovian realizer need not exist. We establish an obstruction based on revival of configuration-space distinguishability, then apply it to Rabi oscillations.

\subsubsection{A distinguishability-revival obstruction}

A Markovian realizer would require stochastic transition matrices between intermediate times satisfying the CK property. Accordingly, we call the classical dynamical family \emph{Markov compatible} on \([s,T]\) if there are column-stochastic matrices \(K(t,u)\) satisfying
\begin{equation}
	\begin{aligned}
		K(u,u)&=I,\\
		K(t,u)K(u,r)&=K(t,r),
		&&s\leq r\leq u\leq t\leq T,\\
		K(t,s)&=\G{t}{s}.
	\end{aligned}
	\label{eq:endpoint-Markov-compatibility}
\end{equation}
No continuity, differentiability, or time homogeneity is assumed.\footnote{We also do not impose the re-rooted matching \(K(t,u)_{ji}=|U(t,u)_{ji}|^2\) used in Gillespie's two-state example~\cite{Gillespie1994}; see the discussion by Garbaczewski and Olkiewicz~\cite{GarbaczewskiOlkiewicz1996}. A distinguishability revival rules out Markovian realizers satisfying rooted matching, and hence also those satisfying the stronger re-rooted matching requirement.}

For two initial distributions \(p,q\), we measure the distinguishability of their configuration distributions at time \(t\) by the total variation distance
\begin{equation}
	D_{p,q}^{(s)}(t)
	:=
	\frac12\left\|\G{t}{s}(p-q)\right\|_1.
	\label{eq:rooted-TV-distinguishability}
\end{equation}

\begin{proposition}[No distinguishability revival]
	\label{prop:no-distinguishability-revival}
	Markov compatibility implies
	\begin{equation}
		D_{p,q}^{(s)}(t)\leq D_{p,q}^{(s)}(u),
		\qquad s\leq u\leq t\leq T,
		\label{eq:TV-monotonicity}
	\end{equation}
	for every \(p,q\). Thus any increase in \(D_{p,q}^{(s)}(t)\) rules out a Markovian realizer on the same configuration space.
	
	In particular, a distinction that has disappeared cannot reappear:
	\begin{equation}
		\G{u}{s}p=\G{u}{s}q
		\quad\Longrightarrow\quad
		\G{t}{s}p=\G{t}{s}q
		\qquad(u\leq t\leq T).
		\label{eq:no-recovery-after-erasure}
	\end{equation}
	
\end{proposition}

\begin{proof}
	The CK property gives
	\(\G{t}{s}=K(t,u)\G{u}{s}\).
	
	Every stochastic matrix is contractive in the $\ell^1$ norm. Indeed, for any real vector $v$,
	\begin{align}
		\|K(t,u)v\|_1
		&=
		\sum_j
		\left|
		\sum_i K(t,u)_{ji}v_i
		\right|                                                     \\
		&\leq
		\sum_{j,i}K(t,u)_{ji}|v_i|                                  \\
		&=
		\sum_i|v_i|
		=
		\|v\|_1,
	\end{align}
	where nonnegativity and unit column sums were used. Apply this inequality to \(v=\G{u}{s}(p-q)\).
\end{proof}

\subsubsection{Counterexample: Rabi oscillations}
We now apply this criterion to a two-state quantum system undergoing Rabi oscillations. We show that the configuration distributions of two initially distinguishable preparations first coincide and then separate again. This revival rules out a Markovian realizer on the same configuration space.

Take \(s=0\) and \(U_X(t)=e^{i\omega t\sigma_X}\), with \(\omega>0\). In the \(\sigma_Z\) configuration basis, 
\begin{equation}
	\G{t}{0}
	=
	\begin{pmatrix}
		\cos^2(\omega t)&\sin^2(\omega t)\\
		\sin^2(\omega t)&\cos^2(\omega t)
	\end{pmatrix}.
	\label{eq:Rabi-shadow}
\end{equation}
For the deterministic preparations \(p=e_1\) and \(q=e_2\),
\begin{equation}
	D_{e_1,e_2}^{(0)}(t)=|\cos(2\omega t)|.
	\label{eq:Rabi-TV}
\end{equation}
At \(t_*=\pi/(4\omega)\), both preparations give
\((1/2,1/2)^{\mathsf T}\), whereas
\[
D_{e_1,e_2}^{(0)}(t_*+\epsilon)
=|\sin(2\omega\epsilon)|>0
\qquad
\left(0<\epsilon<\frac{\pi}{2\omega}\right).
\]
Thus no Markov process can reproduce the rooted Rabi shadow on \([0,T]\) when \(T>t_*\), and therefore any stochastic realizer of this rooted family on the same configuration space must be non-Markovian.

\section{Conclusion}

We have generalized the stochastic--quantum correspondence by allowing arbitrary compatible CPTP lifts on the Hilbert space associated with the given configuration space. The original lifting prescription restricts the density matrices reachable from an initially diagonal state. Our generalized framework overcomes this restriction: as the stochastic kernel ranges over all column-stochastic kernels, its compatible CPTP lifts collectively encompass every quantum channel on the same Hilbert space.

When a compatible  composable family of CPTP lifts exists and satisfies certain regularity assumptions, its evolution is governed by a time-local Lindblad generator. Such a composable description can exist even when the rooted stochastic kernels are indivisible. The off-diagonal degrees of freedom of the lifted state can carry predictive information absent from the instantaneous configuration marginal. The present lifted state can therefore suffice to predict subsequent outcomes even when the present configuration distribution does not.

We established that every rooted quantum dynamical family admits a stochastic realizer on the same configuration space, with conditional path laws that remain fixed as the initial distribution varies. A revival of configuration-space distinguishability rules out every Markovian realizer on that space. Rabi oscillations illustrate that even unitary quantum evolution can therefore require a non-Markovian stochastic realizer.

It is remarkable that quantum mechanics, often regarded as conceptually mysterious, finds a natural home within the conceptually transparent framework of general stochastic processes. The stochastic--quantum correspondence places the probabilities at the heart of quantum theory in a precise mathematical setting, offering a unified framework for interpreting and understanding them. In doing so, it provides a new ontic interpretation of quantum phenomena.

\section*{Acknowledgements}
The author thanks Jacob Barandes for his careful reading of an earlier version of this manuscript, helpful discussions, and corrections and clarifications that have improved this revised version.

\appendix
\renewcommand{\theequation}{\thesection.\arabic{equation}}
\section{The XOR process: CK consistency without Markovianity}
\label{app:xor-process}

Consider a three-time sampling of a two-state stochastic process, with state space \(X=\{0,1\}\), whose joint distribution assigns probability \(1/4\) to each of the following time-ordered trajectories \((x_0,x_1,x_2)\):
\[
(0,0,0),\qquad (0,1,1),\qquad (1,0,1),\qquad (1,1,0),
\]
and probability \(0\) to all other trajectories. We call this the XOR process, since the trajectories of nonzero probability are precisely those satisfying \(x_2 = x_0 \oplus x_1\) so that the final state is given by the exclusive-or of the two earlier ones.

A direct calculation shows that the two-time conditional probabilities are all given by
\begin{equation}
	\bigl[p(x_1\mid x_0)\bigr]
	=
	\bigl[p(x_2\mid x_1)\bigr]
	=
	\bigl[p(x_2\mid x_0)\bigr]
	=
	\frac12
	\begin{pmatrix}
		1 & 1\\
		1 & 1
	\end{pmatrix}.
	\label{eqn:xorextwotimemaps}
\end{equation}
It follows that the two-time transition kernels are divisible at \(t_1\), in the sense that
\begin{equation}
	\Gamma(t_2\leftarrow t_0)
	=
	\Gamma(t_2\leftarrow t_1) \Gamma(t_1\leftarrow t_0).
\end{equation}

However, this does not imply Markovianity, since the future conditional law still depends on the earlier history. In particular,
\begin{equation}
	\bigl[p(x_2\mid x_1,x_0=0)\bigr]
	=
	\begin{pmatrix}
		1 & 0\\
		0 & 1
	\end{pmatrix},
	\qquad
	\bigl[p(x_2\mid x_1,x_0=1)\bigr]
	=
	\begin{pmatrix}
		0 & 1\\
		1 & 0
	\end{pmatrix}.
\end{equation}
Thus,
\begin{equation}
	p(x_2\mid x_1,x_0)\neq p(x_2\mid x_1),
\end{equation}
so the state at \(t_1\) does not screen off the earlier history at \(t_0\). Hence Eq.~\eqref{eqn:markovprop} fails for this process at \(t_1\).

The XOR process is therefore non-Markovian, although its two-time conditional kernels satisfy the CK property at all sampled times. 

\section{A divider distinct from the actual two-time kernel}
\label{app:divider-not-kernel}
Consider a process on \(X=\{0,1\}\) at three times \(t_0<t_1<t_2\). Let \(X_{t_0}\) be a uniformly distributed bit, and let \(B_1,B_2\) be independent bits, independent of \(X_{t_0}\), with \(P(B_1=1)=1/4\) and \(P(B_2=1)=3/8\). Define
\begin{equation}
	X_{t_1}=X_{t_0}\oplus B_1,
	\qquad
	X_{t_2}=X_{t_0}\oplus B_2,
\end{equation}
where \(\oplus\) denotes addition modulo two. The rooted kernels are
\begin{align}
	\Gamma(t_1\leftarrow t_0)
	&=
	\frac14
	\begin{pmatrix}
		3 & 1\\
		1 & 3
	\end{pmatrix},\\
	\Gamma(t_2\leftarrow t_0)
	&=
	\frac18
	\begin{pmatrix}
		5 & 3\\
		3 & 5
	\end{pmatrix}.
\end{align}
Since \(\Gamma(t_1\leftarrow t_0)\) is invertible, the
unique divider is
\begin{equation}
	\widetilde\Gamma(t_2\leftarrow t_1)
	=
	\Gamma(t_2\leftarrow t_0)
	\Gamma(t_1\leftarrow t_0)^{-1}
	=
	\frac14
	\begin{pmatrix}
		3 & 1\\
		1 & 3
	\end{pmatrix}.
\end{equation}
This matrix is column-stochastic, so the rooted family admits a division at \(t_1\).

However, the two later configurations agree precisely when \(B_1=B_2\), which occurs with probability \((3/4)(5/8)+(1/4)(3/8)=9/16\). By symmetry, the actual two-time conditional kernel is therefore
\begin{equation}
	\Gamma(t_2\leftarrow t_1)
	=
	\frac1{16}
	\begin{pmatrix}
		9 & 7\\
		7 & 9
	\end{pmatrix}
	\neq
	\widetilde\Gamma(t_2\leftarrow t_1).
\end{equation}
Thus the divider differs from the actual two-time conditional kernel of the process.

\section{Left--right representation of linear lifts}
\label{app:ABproof}

\subsection{Representation of linear maps}

In finite dimensions, every linear map \(\phi:\mathcal B(\mathcal H)\to\mathcal B(\mathcal H)\) admits a left--right operator-sum representation
\begin{equation}
	\phi(\rho)
	=
	\sum_{\beta=1}^{r} A_\beta \rho B_\beta,
	\qquad r\leq N^2,
	\label{eqn:lineardecomp}
\end{equation}
where \(N=\dim\mathcal H\) and \(A_\beta,B_\beta\in\mathcal B(\mathcal H)\). This representation requires neither positivity nor trace preservation.

To establish Eq.~\eqref{eqn:lineardecomp}, fix the configuration basis \(\{|i\rangle\}_{i=1}^N\) and define the matrix units 
\begin{equation}
E_{ij}:=|i\rangle\langle j|.
\end{equation}
Expand the image of each matrix unit as
\begin{equation}
\phi(E_{kl})
=
\sum_{i,j=1}^N c_{ij,kl}E_{ij}.
\end{equation}
Since \(\rho=\sum_{k,l}\rho_{kl}E_{kl}\), linearity gives
\begin{equation}
\phi(\rho)
=
\sum_{i,j,k,l=1}^N c_{ij,kl}\rho_{kl}E_{ij}.
\end{equation}
Using \(E_{ik}\rho E_{lj}=\rho_{kl}E_{ij}\), this becomes
\begin{equation}
\phi(\rho)
=
\sum_{i,k=1}^N
E_{ik}\rho
\left(\sum_{j,l=1}^N c_{ij,kl}E_{lj}\right).
\end{equation}
Thus the choices
\begin{equation}
A_{ik}:=E_{ik},
\qquad
B_{ik}:=\sum_{j,l=1}^N c_{ij,kl}E_{lj}
\end{equation}
give the required representation after relabelling the pair \((i,k)\) by a single index \(\beta\).

\subsection{Compatibility and stochastic consistency}

For a family of maps \(\phi_{t,s}\) in the form \eqref{eqn:lineardecomp}, the compatibility condition \eqref{eqn:gammatrace} reads
\begin{equation}
	\G{t}{s}_{ji}
	=
	\sum_{\beta=1}^{r}
	\operatorname{Tr}\bigl(
	\hat P^j A_\beta(t,s)\hat P^i B_\beta(t,s)
	\bigr).
	\label{eqn:ABdictcomponent}
\end{equation}
Since
\begin{equation}
\operatorname{Tr}\bigl(
\hat P^j A_\beta\hat P^i B_\beta
\bigr)
=
(A_\beta)_{ji}(B_\beta)_{ij},
\end{equation}
the dictionary can equivalently be written as
\begin{equation}
	\Gamma(t\leftarrow s)
	=
	\sum_{\beta=1}^{r}
	A_\beta(t,s)\odot B_\beta(t,s)^T.
	\label{eqn:ABdict}
\end{equation}

For a prescribed stochastic kernel, Eq.~\eqref{eqn:ABdict} expresses compatibility. Conversely, if the kernel is defined by Eq.~\eqref{eqn:ABdict}, its stochastic-kernel properties are equivalent to the following conditions:

\begin{enumerate}
	\item \textbf{Non-negativity of the kernel entries:}
	\begin{equation}
		\sum_{\beta=1}^{r}
		A_{\beta,ji}(t,s)B_{\beta,ij}(t,s)
		\geq 0
		\qquad \forall i,j.
		\label{eqn:ABpos}
	\end{equation}
	
	\item \textbf{Column normalization:}
	\begin{equation}
		\sum_{\beta=1}^{r}
		\bigl(B_\beta(t,s)A_\beta(t,s)\bigr)_{ii}
		=
		1
		\qquad \forall i.
		\label{eqn:ABtracediag}
	\end{equation}
	
	\item \textbf{Identity at equal times:}
	\begin{equation}
		\sum_{\beta=1}^{r}
		A_{\beta,ji}(s,s)B_{\beta,ij}(s,s)
		=
		\delta_{ji}
		\qquad \forall i,j.
		\label{eqn:ABidentity}
	\end{equation}
\end{enumerate}

These conditions constrain only the diagonal-to-diagonal restriction of the lift.

Trace preservation on the full operator algebra,
\[
\begin{aligned}
	\operatorname{Tr}\bigl(\phi_{t,s}(\rho)\bigr)
	&=
	\operatorname{Tr}\left(
	\left[\sum_{\beta=1}^{r}B_\beta(t,s)A_\beta(t,s)\right]\rho
	\right)\\
	&=
	\operatorname{Tr}(\rho)
	\qquad\text{for all }\rho\in\mathcal B(\mathcal H),
\end{aligned}
\]
is equivalent to the stronger condition
\[
\sum_{\beta=1}^{r}B_\beta(t,s)A_\beta(t,s)
=
I.
\]
\newline
The general left--right representation also includes the CPTP maps discussed in Section~\ref{sec:cptplifts}. Every such map admits a Kraus representation, corresponding to the special choice
\[
A_\beta(t,s)=K_\beta(t,s),
\qquad
B_\beta(t,s)=K_\beta(t,s)^\dagger.
\]
With this identification, Eq.~\eqref{eqn:ABdict} reduces to the CPTP dictionary \eqref{eqn:Kdict}.

The left--right dictionary also includes the conjugation dictionary \eqref{eq:barandes-dictionary} as the case \(r=1\), \(A_1=\theta\), and \(B_1=\theta^\dagger\).

\end{document}